\documentclass[10pt,a4paper,twocolumn]{article}
\usepackage[left=2.00cm, right=2.00cm, top=2.54cm, bottom=1.9cm]{geometry}
\usepackage{graphicx}
\usepackage{natbib}        
\usepackage{amsmath,amssymb}
\usepackage{amsthm}
\usepackage{subcaption}
\usepackage[utf8]{inputenc}
\usepackage[T1]{fontenc}

\newcommand{\tr}[1]{\operatorname{Tr}\left(#1\right)}
\newcommand{\trr}[1]{\operatorname{Tr}^2\left(#1\right)}

\newcommand{\bra}[1]{\langle #1 |}
\newcommand{\ket}[1]{| #1 \rangle}

\newcommand{\CC}{{\mathbb C}}

\newcommand{\cC}{\mathcal{C}}
\newcommand{\cD}{\mathcal{D}}
\newcommand{\cM}{\mathcal{M}}

\newcommand{\cH}{\mathcal{H}}

\newcommand{\gainj}{\tfrac{6c\eta_j\Gamma_j}{2\alpha_j-1}}

\newcommand{\hrho}{\widehat{\rho}}
\newcommand{\Exp}{\mathbb{E}}
\newtheorem{thm}{Theorem}[section]
\newtheorem{lem}{Lemma}[section]
\newtheorem{rem}{Remark}[section]

\date{}

\title{\LARGE \bf Continuous-time Quantum Error Correction with Noise-assisted Quantum Feedback}
\author{Gerardo Cardona$^{1,2}$, Alain Sarlette$^{2,3}$ and Pierre Rouchon$^{1,2}$  
	\thanks{$^{1}$Centre Automatique et Syst\`emes, Mines-ParisTech, PSL Research University. 60 Bd Saint-Michel, 75006 Paris, France.
	}%
	\thanks{$^{2}$QUANTIC lab, INRIA Paris, rue Simone Iff 2, 75012 Paris, France
	}
	\thanks{$^{3}$Electronics and Information Systems Department, Ghent University, Belgium.
	}
	\thanks{{\tt\small gerardo.cardona@mines-paristech.fr, pierre.rouchon@mines-paristech.fr, alain.sarlette@inria.fr}}
}
\begin{document}

\maketitle
\thispagestyle{empty}
\pagestyle{empty}

%
%
%
%
%
%
%
%
			\begin{abstract}
We address the standard quantum error correction using the three-qubit bit-flip code, yet in continuous-time. This entails rendering a target manifold of quantum states globally attractive.  Previous feedback designs could feature spurious equilibria, or resort to discrete kicks pushing the system away from these equilibria to ensure global asymptotic stability. We present a new approach that consists of introducing controls driven by Brownian motions. Unlike the previous methods, the resulting closed-loop dynamics can be shown to stabilize the target manifold exponentially. We further present a reduced-order filter formulation with classical probabilities. The exponential property is important to quantify the protection induced by the closed-loop error-correction dynamics against disturbances. We study numerically the performance of this control law and of the reduced filter.
			\end{abstract}
			
			

%
%

	\section{Introduction}

Developing methods to protect quantum information in the presence of disturbances is essential to improve existing quantum technologies (\cite{reed2012realization,ofekQEC}). Quantum error correction (QEC) codes, encode a logical state into multiple physical states. Similarly to classical error correction, this redundancy allows to protect quantum information from disturbances by stabilizing a \textit{submanifold} of steady states, which represent the nominal logical states \cite{lidar2013quantum,nielsen2002quantum}. As long as a disturbance does not drive the system out of the basin of attraction of the original nominal state, the logical information remains unperturbed. To stabilize the nominal submanifold in a quantum system, a \textit{syndrome diagnosis} stage performs quantum non-destructive (QND) measurements extracting information about code disturbances without perturbing the encoded data. Based on this information, a \textit{recovery} feedback action restores the corrupted state.
QEC is most often presented as discrete-time operations towards digital quantum computing, see e.g.~\cite{nielsen2002quantum}. Not only the design of the underlying control layer, but also the proposal of analog quantum technologies, like solving optimization problems by quantum annealing, motivate a study of QEC in continuous-time, among them \textit{reservoir engineering} and \textit{measurement-based feedback}. 

Reservoir engineering couples the target system to a dissipative ancillary quantum system, such that the entropy introduced by errors on the target system is evacuated through the dissipation of the ancillary one. Reservoir engineering for autonomous QEC has been investigated in \cite{Murch2012BathEng},  \cite{cohen2014AutonomousQEC}, \cite{Guillaud2019CatQubits}. An advantage of this approach is that there is no need for external control logic. However, the challenge is to implement the specific ancillary system and coupling within experimental constraints.

Experimental progress on performing high-fidelity quantum measurements now allows to consider measurement-based feedback in continuous-time. In the context of QEC, this has been addressed in \cite{ahn2002continuous,ahn2003quantum,sarovar2004practical,mabuchi2009continuous}, essentially as proposals illustrated by simulation. The short dynamical timescales of experimental setups is a main difficulty towards implementing complex feedback laws. Furthermore, data acquisition and processing leads to latencies in the feedback loop. This motivates the development of efficiently computable control techniques that are robust against unmodeled dynamics.

In this paper we establish analytical results about the convergence rate of QEC systems towards the nominal submanifold, a prerequisite for analytically quantifying the protection of quantum information. To obtain exponential convergence in a compact space, it is necessary to suppress any spurious unstable equilibria that might remain in the closed-loop dynamics. As we noted in \cite{cardona2018exponential}, this problem is greatly simplified by considering stochastic processes to drive the controls (see also~\cite{zhang2018locally} for  feedback laws with similar stochastic terms). Therefore in the present paper, in the context of QEC, we propose  a \textit{noise-assisted} quantum feedback, acting with Brownian noise whose gain is adjusted in real-time. We show via standard stochastic Lyapunov arguments that this new approach renders the target subspace, containing the nominal encoding of quantum information, globally exponentially stable thanks to feedback from syndrome measurements. Furthermore, our strategy allows to work with a reduced state estimator: while other feedback schemes require to keep track of quantum coherences, the proposed feedback scheme allows for the implementation of a reduced filter that only tracks the populations on the various joint eigenspaces of the measurement operators. 

The paper is organized as follows. Section \ref{section:QEC_OpenLoop} presents the dynamical model of the three-qubit bit-flip code, which is the most basic model in QEC. In section \ref{section:QEC_ClosedLoop} we introduce our approach to feedback using noise and we prove exponential stabilization of the target manifold of the three-qubit bit-flip code. It presents as well the reduced order filter that follows from the feedback scheme. Section \ref{section:Simulations} examines the performance of this feedback and reduced filter to protect quantum information from bit-flip errors.

\begin{rem}{(Stochastic Calculus):}
	We will consider concrete instances of It\={o} stochastic differential equations (SDEs) on $\mathbb{R}^n$ of the form
	\begin{equation}\label{eq:SDE}
	dx=\mu(x)dt+\sigma(x)dW,
	\end{equation}
	where $W$ is a standard Brownian motion on $\mathbb{R}^k$, and $\mu,\sigma$ are regular functions of $x$ with image in $\mathbb{R}^n$ and $\mathbb{R}^{n\times k}$ respectively, satisfying the usual conditions for existence and uniqueness of solutions (\cite[Chapter 3]{khasminskii2011stochastic}) on $\mathcal{S}$, a compact and positively invariant subset of $\mathbb{R}^n$.
	
We will use results on stochastic stability (\cite{khasminskii2011stochastic}). Consider \eqref{eq:SDE} with $\mu(x)=\sigma(x)=0$ for $x \in \mathcal{S}_0 \subset \mathcal{S}$, thus $\mathcal{S}_0$ is a compact set of equilibria. Let $V(x)$, a nonnegative real-valued twice continuously differentiable function with respect to every $x\in \mathcal{S}\setminus \mathcal{S}_0$. Its Markov generator associated with \eqref{eq:SDE} is
	\begin{equation}\label{App:DiffOp}
	\mathcal{A}V=\sum_i \mu_i\frac{\partial}{\partial x_i}V+\frac{1}{2} \sum_{i,j} \sigma_i\sigma_j\frac{\partial^2}{\partial x_ix_j}V,
	\end{equation}
	and
	$$
	\Exp[V(x_t)]=V(x_0)+\Exp\left[{\textstyle \int_0^t} \mathcal{A}V(x_s)ds\right].
	$$
	
\begin{thm}[\cite{khasminskii2011stochastic}] \label{thm:Km}
 If there exists $r>0$ such that $\mathcal{A}V(x)\le -rV(x)$, $\forall x\in \mathcal{S}\setminus\mathcal{S}_0$, then $V(x_t)$ is a supermartingale on $\mathcal{S}$ with exponential decay:
$$\Exp[V(x_t)] \leq V(x_0)\; \exp(-r\, t)\, .$$
\end{thm}

If $V$ is a meaningful way to quantify the distance to a target set $\{ \, x : V(x)=0\, \} \supseteq \mathcal{S}_0$, then this theorem establishes an exponential convergence result in the sense of expectation of $V$. Analysis in the rest of this paper consists in defining a function $V$ and constructing controls that ensure exponential convergence in the above sense.

\end{rem}
\section{Continuous-time dynamics of the three-qubit bit-flip code}\label{section:QEC_OpenLoop}

The general model for a quantum system subject to several measurement channels (see, e.g., \cite{barchielli2009quantum}) is an It\={o} stochastic differential equation of the type
\begin{align}
\label{eq:oldyns} d \rho_t={\textstyle \sum_k} \cD_{L_k}(\rho)dt+\sqrt{\eta_k}\cM_{L_k}(\rho)dW_k\; ,\\
\nonumber dY_k=\sqrt\eta_k\tr{(L_k+L_k^\dagger)\rho}dt+dW_k \; .
\end{align}
We have used the standard super-operator notation $\cD_L(\rho)=\big (L\rho L^\dagger-\tfrac{1}{2}(L^\dagger L \rho + \rho L^\dagger L)\big)$, $\cM_L(\rho)=\big(L\rho+\rho L^\dagger -\tr{\rho(L+L^\dagger)}\rho\big)$, where $L^\dagger$ denotes the complex conjugate transpose of $L$. The state $\rho$ belongs to the set of density matrices $\mathcal{S}=\{\rho\in\CC^{n\times n}: \rho=\rho^\dagger,\rho \text{ positive semidefinite },\tr{\rho}=1 \}$ on the Hilbert space of the system $\cH \simeq \CC^{n\times n}$; the $\{W_k\}$ are independent standard Brownian motions and the $\{dY_k\}$ correspond to the measurement processes of each measurement channel. The $\eta_{k} \in [0,1]$ express the corresponding measurement efficiencies, i.e.~the ratio of the corresponding channel linking the system to the outside world which is effectively captured by the measurement device; channels $k$ with $\eta_k=0$ represent pure loss channels.

The simplest way to model the feedback stage consists in applying an infinitesimal unitary operation to the open-loop evolution, $\rho_{t+dt}=U_t(\rho_t+d\rho_t)U_t^\dagger$, where $U_t=\exp(-i\sum_jH_ju_{t,j} dt)$ with $H_j$ hermitian operators denoting the control Hamiltonians that can be applied, and each $u_{t,j} dt$ a real control input. The fact that $u_{t,j} dt$ may contain stochastic processes requires to treat this feedback action with care, we will come back to this in the next section.

\subsection{Dynamics of the three-qubit bit-flip code}

The three-qubit bit-flip code corresponds to a Hilbert space $\cH=(\CC^2)^{\otimes 3} \simeq \CC^8$, where $\otimes$ denotes tensor product (Kronecker product, in matrix representation). We denote $I_n$ the identity operator on $\CC^n$ and we write $X_k$, $Y_k$ and $Z_k$ the local Pauli operators acting on qubit~$k$, e.g. $X_2=I_2\otimes \sigma_{\! x}\otimes I_2$.  We denote  $\{\ket{0},\ket{1} \}$ the usual basis states, i.e.~the -1 and +1 eigenstates of the $\sigma_z$ operator on each individual qubit (\cite{nielsen2002quantum}).

The encoding on this 3-qubit system is meant to counter bit-flip errors, which map a $\pm 1$ eigenstate of $Z_k$ to the $\mp 1$ eigenstate for each $k=1,2,3$.
The nominal encoding for a logical information 0 (resp.~1) is on the state $\ket{000}$ (resp.~$\ket{111}$). A single bit-flip on e.g.~the first qubit brings this to $X_1 \ket{000} = \ket{100}$ (resp.~$\ket{011}$), which by majority vote can be brought back to the nominal encoding.

In the continuous-time model \eqref{eq:oldyns}, bit-flip errors occurring with a probability $\gamma_k\, dt \ll 1$ during a time interval $[t,t+dt]$ are modeled by disturbance channels, with $L_{k+3} = \sqrt{\gamma_k}\, X_k$ and $\eta_{k+3}=0$, $k=1,2,3$. The measurements needed to implement ``majority vote'' corrections, so-called syndromes, continuously compare the $\sigma_z$ value of pairs of qubits. The associated measurements correspond in \eqref{eq:oldyns} to $L_k = \sqrt{\Gamma_k} \, S_k$ for $k=1,2,3$, with $S_1= Z_2Z_3$, $S_2=Z_1Z_3$, $S_3=Z_1Z_2$ and $\Gamma_k$ representing the measurement strength. This yields the following open-loop model:
{\small
\begin{equation}\label{eq:QEC_OL}
d\rho=\sum_{k=1}^{3}\Gamma_k\cD_{S_k}(\rho)dt+\sqrt{\eta_k\Gamma_k}\cM_{S_k}(\rho)dW_k
+\sum_{s=1}^{3}\gamma_s\cD_{X_s}(\rho)dt.
\end{equation}
}
We further define the operators:
\begin{multline}\label{eq:QEC_proj}
\Pi_{\cC}=\tfrac{1}{4}\big (I_8+\sum_{k=1}^3S_k\big ) \ \text{, } \  \Pi_j:=X_j\Pi_{\cC}X_j, j\in \{1,2,3 \},
\end{multline}
corresponding to orthogonal projectors onto the eigenspaces of the measurement syndromes. $\Pi_{\cC}$ projects onto the nominal code $\cC:=\text{span}(\ket{000},\ket{111})$ (+1 eigenspace of all the $S_k$), whereas $\Pi_j$ projects onto the subspace where qubit $j$ is flipped with respect to the two others. For each $k\in \{C,1,2,3 \}$, we write
$$p_{t,k} := \tr{\Pi_k\rho_t} \; \geq 0$$
the so-called population of subspace $k$, i.e.~the probability that a projective measurement of the syndromes would give the output corresponding to subspace $k$. By the law of total probabilities, $\sum_{k\in \{C,1,2,3 \}} p_{t,k} = 1$ for all $t$.

\subsection{Behavior under measurement only}

We have the following behavior in absence of feedback actions and disturbances.

\begin{lem}\label{prop:OpenLoop}
	Consider \eqref{eq:QEC_OL} with $\gamma_s=0$ for $s\in {1,2,3}$.
	\begin{itemize}				
		\item[(i)] For each $k\in \{C,1,2,3 \}$, the subspace population $p_{t,k}$ is a martingale i.e.~$\Exp(p_{t,k} | p_{0,k}) = p_{0,k}$ for all $t\ge 0$.
		\item[(ii)] For a given $\rho_0$, if there exists $\bar{k}\in \{C,1,2,3 \}$ such that $p_{0,\bar{k}}=1$ and $p_{0,k}=0$ for all $k \neq \bar{k}$, then $\rho_0$ is a steady state of \eqref{eq:QEC_OL}.
		\item[(iii)]  The Lyapunov function
		$$
		V(\rho) =\sum_{k\in\{\cC,1,2,3\}} \sum_{k'\not =k}\sqrt{ p_k p_{k'} }
		$$
		decreases exponentially as	
		$\;\;\Exp [V(\rho_t)] \le e^{-rt}V(\rho_0)\;\;$
		for all $t\geq 0$,  with rate $\;\;
		r=4\;\min_{k\in\{1,2,3 \} }\eta_k\Gamma_k .\;\;
		$
		In this sense the system exponentially approaches the set of invariant states described in point (ii).
	\end{itemize}		    	
\end{lem}
\begin{proof}
	The first two statements are easily verified, we prove the last one. The variables $\xi_j=\sqrt{p_j}$, $j\in\{1,2,3 ,\cC\}$ satisfy the following SDE's:
	\begin{multline*}
	d\xi_\cC=-2 \xi_\cC \,\Big(\sum_{k\in\{1,2,3\}}\eta_k\Gamma_k(1-\xi_\cC^2-\xi_k^2)^2\Big)\, dt \\
	+ 2 \xi_\cC \, \Big( \sum_{k\in\{1,2,3\}} \sqrt{\eta_k\Gamma_k}(1-\xi_\cC^2-\xi_k^2)\,  dW_k  \Big) \; ,
	\end{multline*}
	\begin{multline*}
	d\xi_{j\neq \cC} =-2 \xi_j\, \Big( \eta_j\Gamma_j(1-\xi_\cC^2-\xi_j^2)^2\\ +\sum_{k\in\{1,2,3\}\setminus j} \eta_k\Gamma_k (\xi_\cC^2+\xi_k^2)^2\Big )\, dt\\
	+ 2\xi_j \, \Big(\sqrt{\eta_j\Gamma_j}(1-\xi_\cC^2-\xi_j^2)\,dW_j\\-\sum_{k\in\{1,2,3\}\setminus j} \sqrt{\eta_k\Gamma_k}(\xi_\cC^2+\xi_k^2)\,dW_k\Big) \; ,
	\end{multline*}
	while $V=\sum_{k\in\{\cC,1,2,3\}} \sum_{k'\neq k}\xi_k\xi_{k'}$.
	Noting that $2(1-\xi_\cC^2-\xi_k^2)$ and $2(\xi_\cC^2+\xi_k^2)$ just correspond to $1\pm\tr{\rho S_k}$, we only have to keep track of $\pm$ signs to compute
	\begin{multline*}
	\mathcal{A}V= -2 \sum_{k \in \{\cC,1,2,3\}} \sum_{j \in \{\cC,1,2,,3\} \setminus k} \xi_j \xi_k  \sum_{l \in \{1,2,3\}} \epsilon_{j,k,l} \eta_l \Gamma_l
	\end{multline*}
	where, for each pair $(j,k)$, the selector $\epsilon_{j,k,l} \in \{0,1\}$ equals $1$ for two $l$ values, namely $\epsilon_{\cC,k,l}= \epsilon_{k,\cC,l} = 1$ if $l\neq k \in \{1,2,3\}$ and $\epsilon_{j,k,j}= \epsilon_{j,k,k}=1$ for $j,k \in \{1,2,3\}$. This readily leads to
	$\mathcal{A}V \leq -4\, \min_{k\in\{1,2,3 \} }(\eta_k\Gamma_k) \; V \; .$	
We conclude by Theorem \ref{thm:Km} and noting that $V = 0$ necessarily corresponds to a state as described in point (ii).
\end{proof}

The above Lyapunov function describes the convergence of the state towards $\tr{\Pi_{\bar{k}}\rho}=1$, for a random  subspace $\bar{k} \in \{\cC,1,2,,3\}$ chosen with probability $p_{0,\bar{k}}$. We now address how to render a particular subspace globally attractive, namely the one associated to $\Pi_{\cC}$ and nominal codewords.


\section{Error correction via noise-assisted feedback stabilization}\label{section:QEC_ClosedLoop}

\subsection{Controller design}\label{subsec:controldesign}

Error correction requires to design a control law satisfying two properties:
\begin{itemize}
	\item Drive any initial state $\rho_0$ towards a state with support only on the nominal codespace $\cC=\text{span}\{\ket{000},\ket{111}\}$. This comes down to making $\tr{\Pi_\cC \rho_t}$ converge to $1$.
	\item For $\tr{\Pi_\cC \rho_0}=1$ and in the presence of disturbances $\gamma_s \neq 0$, minimize the distance between $\rho_t$ and $\rho_0$ for all $t\ge 0$.
\end{itemize}
We now directly address the first point, the second one will be discussed in the sequel.

As mentioned in the introduction, this problem has already been considered before, yet without proof of exponential convergence. Towards establishing such proof, we introduce a key novelty into the feedback signal: we drive it by a stochastic process. 
Indeed, noise can be as efficient as a deterministic action to exponentially destabilize a spurious equilibrium where $\bar{k} \neq \cC$; in turn, using noise simplifies the study of the average dynamics, both in the analysis via Theorem \ref{thm:Km} and towards implementing a quantum filter to estimate $\rho$.
We thus introduce \textit{noise-assisted quantum feedback}, where the control input consists of pure noise with state-dependent gain; i.e.~we take 
$$u_jdt=\sigma_j(\rho)dB_j\; ,$$
with $B_j(t)$ a Brownian motion independent of any $W_k(t)$. As control Hamiltonians we take $H_j=X_j$, thus rotating back the bit-flip actions. The closed-loop dynamics in It\={o} sense then writes:
{\small
\begin{multline}\label{eq:QEC_CL}
 d\rho=\sum_{k=1}^{3}\Gamma_k\cD_{S_k}(\rho)dt+\sqrt{\eta_k\Gamma_k}\cM_{S_k}(\rho)dW_k +\sum_{s=1}^{3}\gamma_s\cD_{X_s}(\rho)dt\\
+\sum_{j=1}^3 -i\sigma_j(\rho)[X_j,\rho]dB_j + \sigma_j(\rho)^2\cD_{X_j}(\rho)dt \,.
\end{multline}}
The last term can be viewed as ``encouraging'' a bit-flip with a rate depending on the value of $\sigma_j$ and thus on $\rho$. The remaining task is to design the gains $\sigma_j$.
For this many options will work --- its only essential role is to ``shake'' the state when it is close to $\tr{\Pi_\cC \rho} = 0$, since the open loop already ensures stochastic convergence to either $\tr{\Pi_\cC \rho} = 0$ or $\tr{\Pi_\cC \rho} = 1$.
The following hysteresis-based control law, illustrated by Fig.~\ref{fig:FeedbackZones}, depends only on the $p_{t,k}$ and should not be too hard to implement. Select real parameters $\alpha_j$ and $\beta_j$ such that  $\tfrac{1}{2}<\beta_j<\alpha_j < 1$ for $j\in \{1,2,3 \}$, and take  a constant $c>0$.
\begin{enumerate}
	\item If $ p_{j} \geq \alpha_j$ then take $\sigma_j=\sqrt{\gainj}$; \label{feedback1}
	\item If $ p_{j} \leq \beta_j$  then take $\sigma_j=0$; \label{feedback2}
	\item In the hysteresis region, i.e.~for values of $p_{j} \in ]\beta_j,\alpha_j[$: keep the previous value of $\sigma_j$. \label{feedback3}
\end{enumerate}
\begin{figure}
	\centering
	\includegraphics[width=0.7\linewidth]{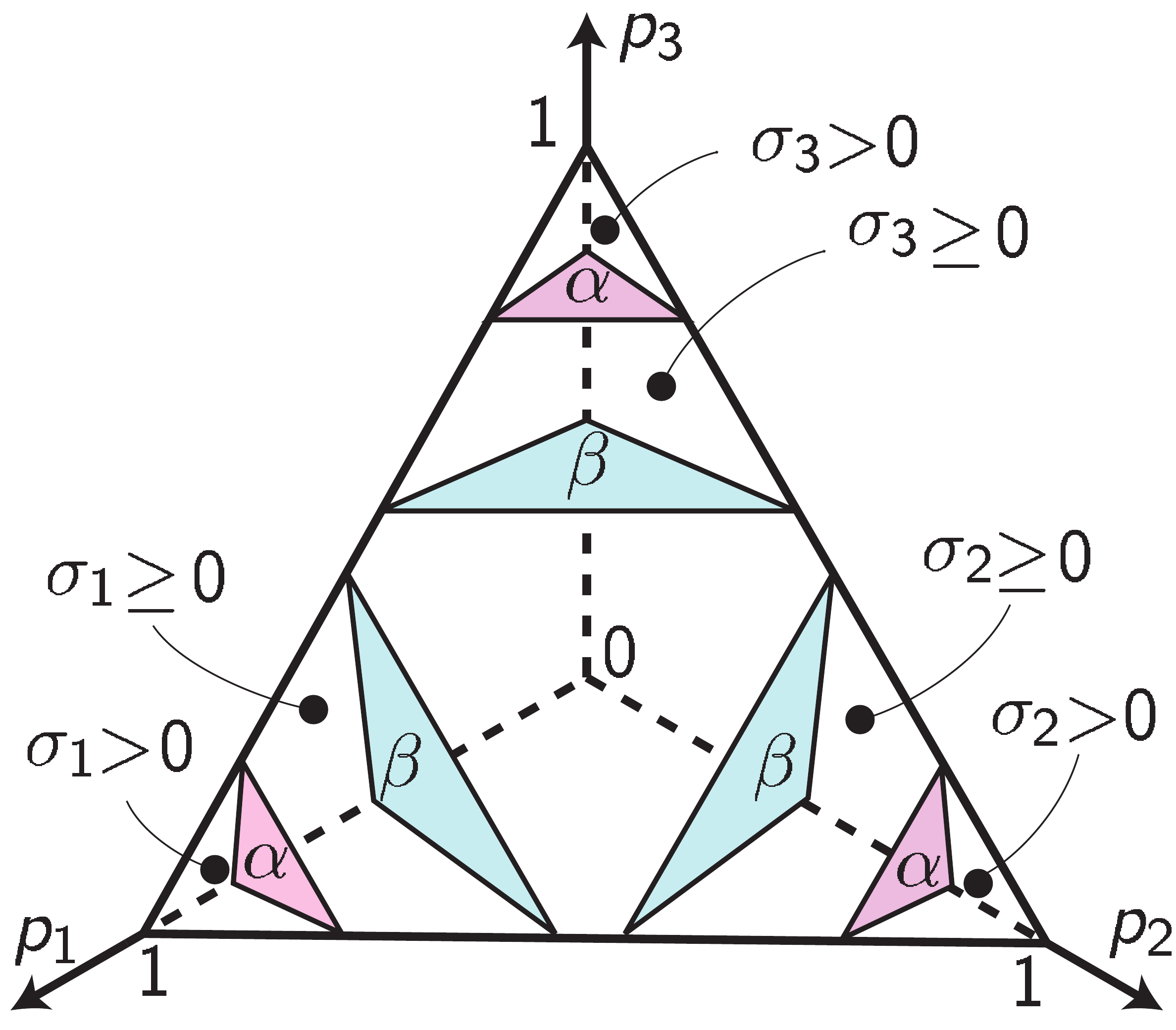}
	\caption{for $\alpha_j\equiv\alpha$ and $\beta_j\equiv\beta$,  the 6 active feedback zones in the simplex {\small$\big\{(p_1,p_2,p_3)~\big|~p_1,p_2,p_3 \geq 0, ~p_1+p_2+p_3 \leq 1 \big\}$}. }
	\label{fig:FeedbackZones}
\end{figure}

\subsection{Closed-loop exponential convergence}\label{ssec:cec}

We propose the closed-loop Lyapunov function:
\begin{equation}\label{eq:V}
  V(\rho)= V_1(\rho) + V_2(\rho) + V_3(\rho)
\end{equation}
with  $ V_k(\rho)= \sqrt{p_k +\; p_1+p_2+p_3}$ for $k=1,2,3$.

\begin{thm}\label{theorem:QEC_Disc}
Consider \eqref{eq:QEC_CL} with all $\gamma_s=0$ and feedback gains  $(\sigma_j)$ as specified just before section \ref{ssec:cec}. Then
	\begin{equation*}
	\Exp[V(\rho_t)]\le V(\rho_0) e^{-rt}, \; \forall t\ge 0,
	\end{equation*}
	with the exponential convergence rate estimated as:
{\small
\begin{multline*}
  	r= \left(\min_{j\in\{1,2,3\}}\!\! \eta_j\Gamma_j\right)
          \min\left( c~ ,~  \tfrac{4}{3\sqrt{2}} \min_{(s,x_1,x_2,x_3)\in K } g(s,x_1,x_2,x_3)\right)
\end{multline*}
}
where  $g(s,x_1,x_2,x_3)$ is given in~\eqref{eq:gg} below and
{\small
$$
K=\Big\{(s,x_1,x_2,x_3)\in[0,1]^4~\Big|~x_1+x_2+x_3=1; sx_j \leq \alpha_j \Big\}
$$
}
\end{thm}
\begin{proof}
By design of the hysteresis, well-posedness of the solution then follows from standard arguments on the construction of solutions of SDE's. The proof then consists in showing that $V(\rho_t)$ on $\mathcal{S}$ is an exponential supermartingale satisfying $\mathcal{A}(V) \leq  - r V$. Towards this we partition the state-space into $\mathcal{Q}:=\cup_{j=1}^3\big\{\rho\in\mathcal{S}~|~ p_j\geq \alpha_j\big\}$ and $\mathcal{S}\setminus\mathcal{Q}$ on which we compute the Markov generator $\mathcal{A}(V)$ separately.
	
Let us write from~\eqref{eq:QEC_CL} the expression of $\mathcal{A}V(\rho)= \Exp\Big[ dV_t ~|~\rho_t=\rho\Big]/dt$  for any value of the control gain vector $\sigma$. We exploit here the following formula based on It\={o} rules and valid for any non-negative operator $F$:
{\small
 $$
 d \sqrt{\tr{F\rho}}  =\frac{\tr{F\,d\rho}}{2\sqrt{\tr{F\rho}}} - \frac{(\tr{F\,d\rho})^2}{4\tr{F\rho} \sqrt{\tr{F\rho}}}
.
 $$
}
We detail below the computations when $\eta_j\equiv\eta$ and $\Gamma_j\equiv \Gamma$ (the formulas in the general case are slightly more complicated).
With $F_1=2\Pi_1+\Pi_2+\Pi_3$ and $V_1(\rho)=\sqrt{f_1}=\sqrt{\tr{F_1\rho}}$, we get
{\small
\begin{multline*}
   \mathcal{A}V_1(\rho) = \tfrac{2\sigma_1^2\big(1-f_1\big) + \sigma_2^2\big(1-2(p_1+p2)\big)+ \sigma_3^2\big(1-2(p_1+p_3)\big)}{2\sqrt{f_1}}
   \\
 \hspace{-1em}  - 4\eta\Gamma\tfrac{\big((p_2+p_3)(1-f_1)\big)^2+ \big(p_1+(p_1+p_3)(1-f_1)\big)^2+ \big(p_1+(p_1+p_2)(1-f_1)\big)^2 }{f_1 \sqrt{f_1}}
   \\
   - \tfrac{\sigma_1^2\trr{[X_1,\rho]F_1}+\sigma_2^2\trr{[X_2,\rho]F_1}+\sigma_3^2\trr{[X_3,\rho]F_1}}{4f_1 \sqrt{f_1}} \; .
\end{multline*}
}
Since $\sqrt{f_1} \geq \tfrac{1}{3\sqrt{2}} V$, we have
{\small
\begin{multline*}
   \mathcal{A}V_1(\rho) \leq  \tfrac{2\sigma_1^2\big(1-f_1\big) + \sigma_2^2\big(1-2(p_1+p2)\big)+ \sigma_3^2\big(1-2(p_1+p_3)\big)}{2\sqrt{f_1}}
   \\
 -    \tfrac{4\eta\Gamma V\big((p_2+p_3)(1-f_1)\big)^2+ \big(p_1+(p_1+p_3)(1-f_1)\big)^2+ \big(p_1+(p_1+p_2)(1-f_1)\big)^2 }{3\sqrt{2} f^2_1}
   .
\end{multline*}
}
Via circular permutation and summation, we get
\begin{equation}\label{eq:AV}
  \mathcal{A}V(\rho) \leq \sum_{j=1}^{3} \sigma_j^2(\rho) g_j(\rho) -  \tfrac{4\eta\Gamma}{3\sqrt{2}} g(\rho) V(\rho)
\end{equation}
where
$\;
g_j(\rho) = \tfrac{1-f_j}{\sqrt{f_j}} + \tfrac{1-2(p_j+p_{j'})}{2\sqrt{f_{j'}}} + \tfrac{1-2(p_j+p_{j''})}{2\sqrt{f_{j''}}}
\;$
with $\{j,j',j''\}=\{1,2,3\}$ and $\;  g(\rho) =  $
\begin{multline*}
     \tfrac{\big((p_2+p_3)(1-f_1)\big)^2+ \big(p_1+(p_1+p_3)(1-f_1)\big)^2+ \big(p_1+(p_1+p_2)(1-f_1)\big)^2 }{(2p_1+p_2+p_3)^2}
  \\
  + \tfrac{\big((p_3+p_1)(1-f_2)\big)^2+ \big(p_2+(p_2+p_1)(1-f_2)\big)^2+ \big(p_2+(p_2+p_3)(1-f_2)\big)^2 }{(2p_2+p_3+p_1)^2}
  \\
  + \tfrac{\big((p_1+p_2)(1-f_3)\big)^2+ \big(p_3+(p_3+p_2)(1-f_3)\big)^2+ \big(p_3+(p_3+p_1)(1-f_3)\big)^2 }{(2p_3+p_1+p_2)^2}
  .
\end{multline*}
When $\rho\in\mathcal{Q}$, we have $p_j \geq \alpha_j > 1/2$ for a unique $j\in\{1,2,3\}$, since $p_1+p_2+p_3 \leq 1$. Assume first that $p_1\geq \alpha_1$, thus $\sigma_1=\sqrt{ \tfrac{6c\eta\Gamma}{2\alpha_1-1}}$ and $\sigma_2(\rho)=\sigma_3(\rho)=0$. Since $g(\rho)\geq 0$, inequality~\eqref{eq:AV} implies
$$
\mathcal{A} V \leq  \tfrac{6c\eta\Gamma}{2\alpha_1-1} \left(\tfrac{1-f_1}{\sqrt{f_1}} + \tfrac{1-2(p_1+p_{2})}{2\sqrt{f_{2}}} + \tfrac{1-2(p_1+p_{3})}{2\sqrt{f_{3}}}\right)
.
$$
Since $f_1\geq 2\alpha_1$, $1-2p_1\leq 0$, $f_1\leq 2$ and   $V \leq 3 \sqrt{2}$ we get
$$
\mathcal{A} V \leq   \tfrac{6c\eta\Gamma}{2\alpha_1-1}  \tfrac{1-2\alpha_1}{\sqrt{f_1}}= -\tfrac{6c\eta\Gamma}{V\sqrt{f_1}}   V
\leq - c \eta \Gamma  V
.
$$
We get a similar inequality when $p_2\geq \alpha_2$ or  $p_3 \geq \alpha_3$. Thus
$$
\forall \rho\in \mathcal{Q},~ \mathcal{A} V(\rho) \leq  - c \eta \Gamma  V(\rho)
.
$$
Consider now $\rho\in\mathcal{S}\setminus\mathcal{Q}$. Then, $p_j < \alpha_j$ for all $j$. Since  $\sigma_j(\rho)=0$ when $p_j \leq 1/2$ we have
$\sigma_j^2(\rho)  g_j(\rho) \leq 0$. From~\eqref{eq:AV}, we have   $\mathcal{A} V(\rho) \leq  -  \tfrac{4\eta\Gamma}{3\sqrt{2}} g(\rho) V(\rho)$.
Let us prove that $g(\rho) \geq r$ for any $\rho\in\mathcal{S}\setminus\mathcal{Q}$.  With $s=p_1+p_2+p_3$ and $x_j=p_j/s$, $g$ can be seen as a function of $(s,x_1,x_2,x_3)$,
{\small
\begin{multline} \label{eq:gg}
 g(\rho) =  g(s,x_1,x_2,x_3)\triangleq\\
     \tfrac{\big((x_2+x_3)(1-f_1)\big)^2+ \big(x_1+(x_1+x_3)(1-f_1)\big)^2+ \big(x_1+(x_1+x_2)(1-f_1)\big)^2 }{(1+x_1)^2}
  \\
  + \tfrac{\big((x_3+x_1)(1-f_2)\big)^2+ \big(x_2+(x_2+x_1)(1-f_2)\big)^2+ \big(x_2+(x_2+x_3)(1-f_2)\big)^2 }{(1+x_2)^2}
  \\
+  \tfrac{\big((x_1+x_2)(1-f_3)\big)^2+ \big(x_3+(x_3+x_2)(1-f_3)\big)^2+ \big(x_3+(x_3+x_1)(1-f_3)\big)^2 }{(1+x_3)^2}
\end{multline}}
with $f_j= 1 - s -sx_j$.  Here $(s,x_1,x_2,x_3)$ belongs to the compact set $s\in[0,1]$, $x_j\geq 0$, $\sum_j x_j=1$ and $s x_j \leq \alpha_j$ for all $j$. On this compact set, $g$ is a smooth function. Moreover it is strictly positive since $g=0$ implies that $s=1$ and $x_j=1$ for some $j\in\{1,2,3\}$ which would not satisfy $s x_j \leq \alpha_j$.  This means that  $\min_{\rho\in\mathcal{S}\setminus\mathcal{Q}} g(\rho) >0$.

Taking all things together, we have proved that $\mathcal{A}V(\rho) \leq -r V(\rho)$ always holds. We conclude with  Theorem \ref{thm:Km}.
\end{proof}

For  a heuristic  estimate of $r$,  take $s=\alpha_j$  with $x_j=1$ for some $j$ to get
$
r \sim \left(\min_{j\in\{1,2,3\}} \eta_j\Gamma_j \right)~ \min\left( c, \tfrac{8}{\sqrt{2}}(1-\bar \alpha)^2\right)
$,
with $\bar \alpha= \max_{j\in\{1,2,3\}} \alpha_j$.  Typically one would take $c=1$ and  $\alpha_1=\alpha_2=\alpha_3=\alpha$  close to 1. When  $\eta_{j}\Gamma_{j}$ are all equal, such a rough  estimate  simplifies to $r=4\sqrt{2}(1-\alpha)^2 \eta \Gamma \, .$

\subsection{Reduced quantum filter}
Towards implementing the control law we have to reconstruct in real-time the quantum state estimate $\rho$ via a quantum filter. For~\eqref{eq:QEC_CL}, this filter reads:
\begin{multline}\label{eq:FullFilter}
d\rho=\sum_{k=1}^{3}\Gamma_k\cD_{S_k}(\rho)dt+\sum_{s=1}^{3}\gamma_s\cD_{X_s}(\rho)dt\\+\sum_{k=1}^{3}\sqrt{\eta_k\Gamma_k}\cM_{S_k}(\rho) \Big(dY_k - 2 \sqrt{\eta_k \Gamma_k} \tr{S_k \rho }dt \Big) \\
+\sum_{j=1}^3 -i\sigma_j(\rho)[X_j,\rho]dB_j + \sigma_j(\rho)^2\cD_{X_j}(\rho)dt \,.
\end{multline}
where $dY_k = 2 \sqrt{\eta_k \Gamma_k} \tr{S_k \rho} dt + dW_k$ is the measurement outcome of syndrome $S_k$, and the random $dB_j$ applied to the system are accessible too a posteriori.

Instead, we can replace the state $\rho_t$ in the feedback law, by  $\hrho_t$ corresponding to the Bayesian estimate of $\rho_t$  knowing its initial condition $\rho_0$ and the syndrome measurements $dY_k$ between $0$ and the current time $t>0$, but not the $dB_j$. Then $\hrho_t$ obeys to the SME:
\begin{multline}\label{eq:ReducedFilter}
d\hrho=\sum_{k=1}^{3}\Gamma_k\cD_{S_k}(\hrho)dt+\sum_{j=1}^{3}(\gamma_s+\sigma_j^2(\hrho)) \cD_{X_j}(\hrho)dt
\\+\sum_{k=1}^{3}\sqrt{\eta_k\Gamma_k}\cM_{S_k}(\hrho) \big(dY_k - 2 \sqrt{\eta_k \Gamma_k} \tr{S_k \hrho}dt\big) 
\end{multline}
where $dY_k= 2 \sqrt{\eta_k \Gamma_k} \tr{S_k \rho }dt + dW_k$ with $\rho$ governed by~\eqref{eq:QEC_CL} where $\sigma_j(\rho)$ is replaced by $\sigma_j(\hrho)$.
Denote $\hat p_j= \tr{\Pi_j \hrho}$ and $\hat s_k=\tr{S_k \hrho}$. Then we have
\begin{multline}\label{eq:Qfilter}
  d\hat s_1 = -2 (\gamma_2 + \sigma_2^2+\gamma_3+\sigma_3^3) \hat s_1 dt
  \\
  +2 \sqrt{\eta_1 \Gamma_1} (1-\hat s_1^2) \big(dY_1 - 2 \sqrt{\eta_1 \Gamma_1}\hat s_1 dt\big)
  \\
    +2 \sqrt{\eta_2 \Gamma_2}(\hat s_3 - \hat s_1 \hat s_2)  \big(dY_2 - 2 \sqrt{\eta_2 \Gamma_2}\hat s_2 dt\big)
\\
  +2 \sqrt{\eta_3 \Gamma_3}(\hat s_2 - \hat s_1 \hat s_3) \big(dY_3 - 2 \sqrt{\eta_3 \Gamma_3} \hat s_3 dt\big)
\end{multline}
with $\hat p_1= (1 + \hat s_1 -\hat s_2 - \hat s_3)/4$. The formulas for
 $d\hat s_{2,3}$ and   $\hat p_{2,3}$  are obtained via circular permutation in $\{1,2,3\}$.
Since the feedback law depends only on the populations $\hat{p}_j$, it can be implemented with the exact quantum filter reduced to $(\hat s_1,\hat s_2,\hat s_3) \in \mathbb{R}^3$. Contrarily to the full quantum filter~\eqref{eq:FullFilter}, here the syndrome dynamics $\hat s_k$ are independent of any coherences among the different subspaces and we get a closed system on classical probabilities, driven by the measurement signals.

\section{On the protection of quantum information}\label{section:Simulations}

It is well-known in control theory that exponential stability gives an indication of robustness against unmodeled dynamics. In the present case, this concerns the first control goal, namely stabilization of $\rho_t$ close to the nominal subspace $\cC$ in the presence of bit-flip errors $\gamma_s \neq 0$. About the second control goal, namely keeping the dynamics on $\cC$ close to zero such that logical information remains protected, the analysis of the previous section is less telling.

We can illustrate both control goals by simulation. As in~\cite{ahn2002continuous} we set as initial condition $\rho_0=\ket{000}\bra{000}$ and simulate 1000 closed-loop trajectories under the feedback law of section~\ref{subsec:controldesign}. We compare the average evolution of this encoded qubit with a single physical qubit subject to a $\sigma_x$ decoherence of the same strength, since this is the situation that the bit-flip code is meant to improve. Parameter values and simulation results are shown on Figure~\ref{fig:BitFlipIdeal} where we consider that the quantum filter perfectly follows \eqref{eq:QEC_CL}. Figure~\ref{fig:BitFlipReal} corresponds to a more realistic situation where the same feedback law relies on  the reduced order quantum filter~\eqref{eq:Qfilter} corrupted by errors and feedback latency: we observe a small change of performance but still a clear improvement compared to a single qubit.

\begin{figure}[htb]
	\centering
	\includegraphics[width=\linewidth]{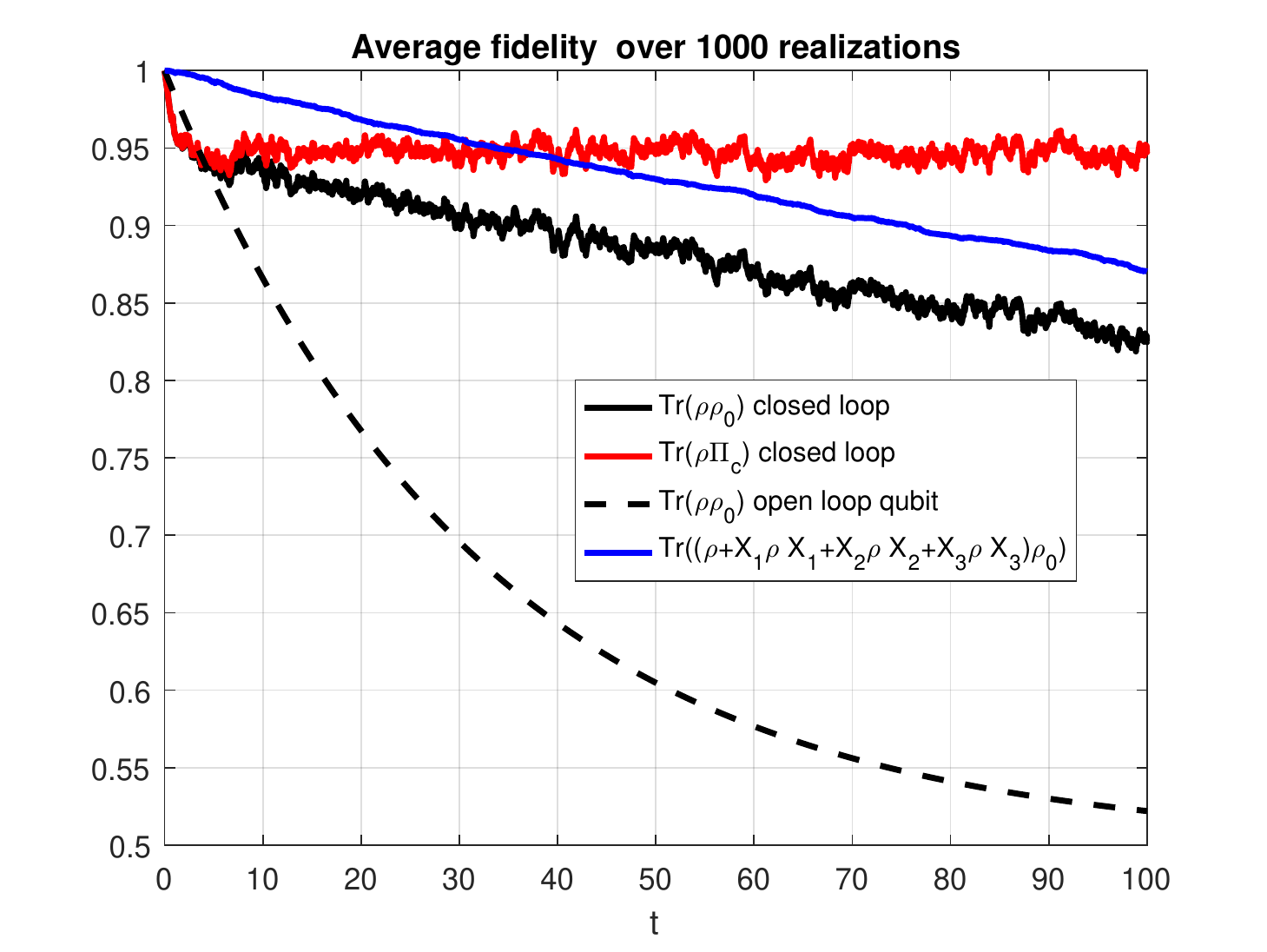}
	\caption{Ideal situation where the feedback of subsection~\ref{subsec:controldesign} is based on $\rho$ governed by~\eqref{eq:QEC_CL}. Solid red: mean overlap of the state with the code space. Solid  black: mean fidelity of the logical qubit versus  $\rho_0$. Solid blue: mean correctable  fidelity  under active quantum feedback. For the three solid curves, the initial state is chosen as $\rho_0=\ket{000}\bra{000}$ and closed-loop simulation parameters based on~\eqref{eq:QEC_CL}  are $\Gamma_j=1$, $\gamma_j=1/64$, $\eta_j=0.8$, and for the feedback law $\beta_j=0.6$, $\alpha_j=0.95$, $c=3/2$.
Dashed line, for comparison: mean fidelity towards $\ket{0}\bra{0}$ for a single physical qubit without measurement nor control and subject to bit-flip disturbances with $\gamma=1/64$.  }
	\label{fig:BitFlipIdeal}
\end{figure}

\begin{figure}[htb]
	\centering
	\includegraphics[width=\linewidth]{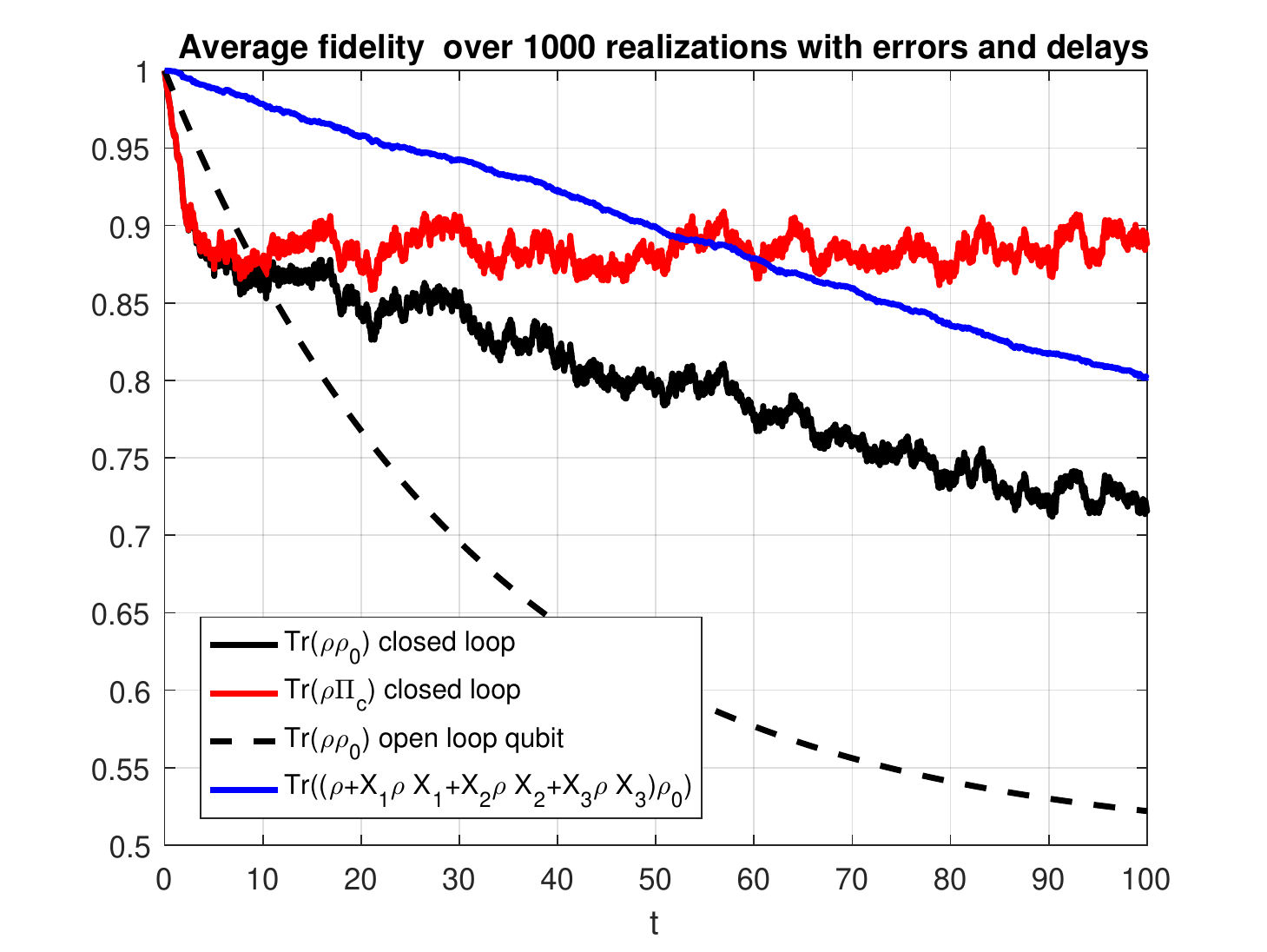}
	\caption{Simulation similar to Fig.~\ref{fig:BitFlipIdeal} for a more realistic case where feedback is based on the reduced order  filter~\eqref{eq:Qfilter} and includes modeling/measurement errors and feedback latency. Marked  with subscript $*$, the parameter values used in~\eqref{eq:Qfilter} are as follows: $\gamma_*=0.8\gamma$, $\Gamma_*=0.9 \Gamma$, $\eta_*=0.9\eta$; constant measurement bias according to  $dY_{*,1}=dY_1+\tfrac{\sqrt{\eta\Gamma}}{10} dt $, $dY_{*,2}=dY_2-\tfrac{\sqrt{\eta\Gamma}}{10} dt$ and  $dY_{*,3}=dY_3+\tfrac{\sqrt{\eta\Gamma}}{20} dt$, , and feedback latency of $1/(2\Gamma)$;  measurement signals  $Y_k$ are based on~\eqref{eq:QEC_CL} with nominal values identical to simulation of Fig.~\ref{fig:BitFlipIdeal}  and control values  $\sigma_j(\hrho)$.
}
	\label{fig:BitFlipReal}
\end{figure}

Regarding the first control goal, we observe that the controller indeed confines the mean evolution to a small neighborhood of $\cC$, for all times, as expected from our analysis. Regarding the second criterion, the distance between $\rho_t$ and $\rho_0$ cannot be confined to a small value for all times. Indeed, majority vote can decrease the rate of information corruption but not totally suppress it; as corrupted information is irremediably lost, $\rho_t$ progressively converges towards an equal distribution of logical 0 and logical 1. However, for the protected 3-qubit code, this information loss is much slower than for the single qubit; this indicates that the 3-qubit code with our feedback law indeed improves on its components.

In our feedback design, making $\alpha_j$ closer to $1/2$ would improve the convergence rate estimate in Theorem \ref{theorem:QEC_Disc}; however, this also has a negative effect on the logical information, since it means that we turn on the noisy drives more often. Analytically computing the optimal tradeoff is the subject of ongoing work. Similary, making $c$ larger would accelerate the recovery action but increase the level of noise, and we want to keep the induced motion slower than the measurement timescale. Simulations (not reproduced here) clearly show that intermediate values of the control parameters deliver better overall results.

\section{Conclusion}
We have approached continuous-time quantum error correction in the same spirit as \cite{ahn2002continuous}, and showed how introducing Brownian motion to drive control fields yields exponential stabilization of the nominal codeword manifold. The main idea relies on the fact that the SDE in open loop stochastically converges to one of a few steady-state situations, but on the average does not move closer to any particular one. It is then sufficient to activate noise only when the state is close to a bad equilibrium, in order to induce globale convergence to the target ones.  This general idea can be extended to other systems with this property, and in particular to more advanced error-correcting schemes. In the same line, while we have  proposed particular controls with hysteresis, proving a similar property with smoother control gains should not be too different.
 The convergence rate obtained is dependent on our choice of Lyapunov function and on the values of $\alpha_j$; from parallel investigation it seems possible to get a closed-loop convergence rate arbitrarily close to the measurement rate.

However, unlike in classical control problems, the key performance indicator is not how fast we approach the target manifold. Instead, what matters is how well, in presence of disturbances, we preserve the encoded information. Towards this goal, we should refrain from disturbing the system with feedback actions; accordingly, we have noticed that taking $\alpha_j$ closer to 1 can improve the codeword fidelity, despite leading to a slower convergence rate estimate. A theoretical analysis of information protection capabilities is the subject of ongoing work.\\

The authors would like to thank K. Birgitta Whaley and Leigh S. Martin for discussions on continuous-time QEC. This work has been supported by the ANR project HAMROQS.


\end{document}